\documentclass[aps,prb,twocolumn,superscriptaddress,showpacs,floatfix]{revtex4-2}
\usepackage{amsmath}
\usepackage{amssymb}
\usepackage{amsfonts}
\usepackage{amsthm}
\usepackage[dvips]{graphicx}
\usepackage{hyperref}
\hypersetup{colorlinks=true,linkcolor=blue,citecolor=blue,urlcolor=blue}

\theoremstyle{plain}
\newtheorem{theorem}{Theorem}

\usepackage{bm}
\usepackage{braket}
\usepackage{mathtools}

\begin{document}

\title{Out-of-time-order correlator computation based on discrete truncated Wigner approximation}

\author{Tatsuhiko Shirai}
\email{tatsuhiko.shirai@aoni.waseda.jp}
\affiliation{Waseda Institute for Advanced Study, Waseda University, Nishi Waseda, Shinjuku-ku, Tokyo 169-0051, Japan}

\author{Takashi Mori}
\affiliation{Department of Physics, Keio University, Kohoku-ku, Yokohama, Kanagawa 223-8522, Japan}

\date{\today}

\begin{abstract}
    We propose a method based on the discrete truncated Wigner approximation (DTWA) for computing out-of-time-order correlators.
    This method is applied to long-range interacting quantum spin systems where the interactions decay as a power law with distance.
    As a demonstration, we use a squared commutator of local operators and its higher-order extensions that describe quantum information scrambling under Hamilton dynamics.
    Our results reveal that the DTWA method accurately reproduces the exact dynamics of the average spreading of quantum information (i.e., the squared commutator) across all time regimes in strongly long-range interacting systems.
    We also identify limitations in the DTWA method when capturing dynamics in weakly long-range interacting systems and the fastest spreading of quantum information.
    Then we apply the DTWA method to investigate the system-size dependence of the scrambling time in strongly long-range interacting systems.
    We reveal that the scaling behavior of the scrambling time for large system sizes qualitatively changes depending on the interaction range.
    This work provides and demonstrates a new technique to study scrambling dynamics in long-range interacting quantum spin systems.
\end{abstract}

\maketitle

\section{Introduction}
Out-of-time-order correlator (OTOC)~\cite{larkin1969quasiclassical} has attracted attention in non-equilibrium statistical mechanics, quantum information, and quantum gravity~\cite{swingle2018unscrambling}.
The time evolution of the OTOC estimates the scrambling time $t_*$ when local perturbation propagates to the entire system.
Systems that exhibit logarithmic scrambling time, $t_* \sim \beta \hbar \log N$, where $\beta$ is the inverse temperature, $\hbar$ is the Planck constant, and $N$ is the number of degrees of freedom, are referred to as fast scramblers~\cite{maldacena2016bound,bentsen2019fast}, with a holographic duality to black hole being explored~\cite{shenker2014black}.
Meanwhile, ballistic information propagation in a chaotic spin system with short-range interaction~\cite{lieb1972finite} implies a polynomial scrambling time, $t_* \sim N^{1/d}$, where $d$ is the space dimension.
Coherently simulating a reverse time evolution has facilitated the measurements of the OTOC in trapped ions~\cite{garttner2017measuring,tian2022testing}, nuclear magnetic resonance systems~\cite{li2017measuring}, and superconducting circuits~\cite{braumuller2022probing}.


Long-range interacting quantum spin systems, where the interaction decays as a power law $J \sim r^{-\alpha}$ with distance $r$, exhibit intriguing dynamical properties~\cite{defenu2023long}.
Rigorous bounds on the scrambling time indicate that translationally invariant spin chains with extensive energy are not fast scramblers, i.e. $t_* \sim N^\gamma$~\cite{tran2020hierarchy,yin2020bound,kuwahara2021absence}.
Although the scrambling dynamics were numerically investigated in previous works~\cite{colmenarez2020lieb, richter2023transport, qi2023surprises}, accurately estimating the exponent $\gamma$ is challenging due to the exponentially growing Hilbert space and the strong finite-size effects.
Therefore, developing approximated methods to tackle this problem is crucial.


Existing approximated methods are insufficient to characterize scrambling dynamics.
A classical approach (as defined in Sec.~\ref{secA:result}) completely neglects the effects of quantum fluctuations on temporal correlation functions.
Truncated Wigner approximation~\cite{polkovnikov2010phase}, which relies on a continuous Wigner function defined over a continuous phase space, fails to capture the spatial structure of correlations.
Since scrambling dynamics is intrinsically tied to the propagation of local correlations throughout the system, a method capable of resolving the dynamics at the level of individual spins is necessary.

One promising approach to address these problems is the discrete truncated Wigner approximation (DTWA)~\cite{schachenmayer2015many-body}.
DTWA offers a semiclassical method that enables efficient simulation of quantum dynamics by representing quantum states in a discrete phase space~\cite{william1987wigner}.
DTWA has been shown to accurately reproduce collective observables, spatial correlation functions, and relative entropy in the long-range interacting systems~\cite{schachenmayer2015many-body,mori2019prethermalization,kunimi2021performance}.
However, DTWA-based approaches for computing time-correlation functions have not yet been developed. 

In this study, we propose a DTWA-based method for computing temporal correlation functions including OTOCs.
The method is applied to calculate a squared commutator of local operators, along with its higher-order extensions, to describe quantum information spreading under the Hamilton dynamics.
We benchmark the DTWA method across systems with varying values of $\alpha$, covering from weakly $(\alpha > d)$ to strongly $(\alpha \leq d)$ long-range interactions.
In the strongly long-range interacting regime, unlike the weakly long-range case, the system exhibits nonadditivity in its energy. 
Our numerical results reproduce the exact dynamics for the average spreading of quantum information (i.e., squared commutators) across all time regimes in strongly long-range interacting systems.
Furthermore, we identify the limitations of the DTWA method when simulating scrambling dynamics in weakly long-range interacting systems and the fastest spreading of quantum information.
In addition, we observe that the approximated dynamics of an autocorrelation function holds valid over short-time scales, with these timescales being shorter as $\alpha$ increases.
These findings clarify the applicability range of the DTWA method in exploring scrambling dynamics in long-range interacting quantum spin systems.
Finally, we apply the DTWA method to investigate the system-size dependence of the scrambling time in strongly long-range interacting regime and compare the numerical result with a theoretical lower bound. 
We reveal that the scaling behavior of the scrambling time for large system sizes qualitatively changes depending on the value of $\alpha$.

This paper is organized as follows. 
In Sec.~\ref{sec:model_method}, we present a model for long-range interacting systems, the DTWA method for computing time-correlation functions including OTOCs, and its benchmark results.
In Sec.~\ref{sec:result}, we provide a theoretical bound and the numerical results using the DTWA approach for the scrambling time.
Section~\ref{sec:conclusion} summarizes this paper with some future directions.
Appendix provides the derivation of the DTWA expression for OTOCs, efficient exact simulation method at $\alpha=0$, and system-size dependences of the DTWA method.

\section{Model and Methods}\label{sec:model_method}
\subsection{Long-range interacting quantum spin systems}
We consider a quantum spin system on a lattice $\Lambda \in \{1, \ldots, N\}$.
The Hamiltonian is given by
\begin{equation}
\hat{H}=\sum_{\substack{i,j \in \Lambda\\ (j>i)}} \hat{\bm{\sigma}}_i \bm{J}_{ij} \hat{\bm{\sigma}}_j^\intercal + \sum_{i\in \Lambda} \bm{h} \hat{\bm{\sigma}}_i^\intercal,
\label{eq:model}
\end{equation}
where $\bm{\sigma}_i=(\hat{\sigma}_i^x, \hat{\sigma}_i^y, \hat{\sigma}_i^z)$ is the vector for Pauli spin operators acting on site $i \in \Lambda$ and $\intercal$ denotes the transpose.
$\bm{J}_{ij} = \{ J_{ij}^{ab} \}_{a,b\in \{x,y,z\}}$ represents the interaction matrix between spins $i$ and $j$, with $J_{ij}^{ab} = J_{ji}^{ba}$, and $\bm{h}=(h^x,h^y,h^z)$ describes a uniform magnetic field, respectively.
The interaction strength decays as a power of $\alpha$ with distance $r_{ij}$, given as
\begin{equation}
    J_{ij}^{ab}=\frac{J^{ab}}{\mathcal{N}(\alpha)}r_{ij}^{-\alpha}.
    \label{eq:coupling}
\end{equation}
Although the proposed DTWA method is applicable regardless of boundary conditions and lattice topologies, we herein assume a one-dimensional lattice with a periodic boundary condition.
Then $r_{ij} = \min \{|i-j|, N-|i-j|\}$.
The interaction is called strongly long range when $0\leq\alpha \leq 1$, whereas weakly long range when $\alpha > 1$~\cite{defenu2023long}.
The normalization of $\mathcal{N}(\alpha)=\sum_{i=2}^N r_{1i}^{-\alpha}$ is known as the Kac prescription~\cite{kac1963vanderWaals} so that the energy per spin is finite in the thermodynamic limit even at $0\leq \alpha \leq 1$.
The model is reduced to an infinite-range model at $\alpha=0$, whereas a short-range model with nearest-neighbor interaction at $\alpha \to \infty$.
In numerical simulations, we adopt $J^{ab}=\delta_{az}\delta_{bz}$ and $(h^x,h^y,h^z)=(0.9045, 0, 0.809)$, where the eigenstate thermalization hypothesis was numerically shown in the limit of $\alpha \to \infty$~\cite{kim2014testing}.
Here, $\delta_{ab}$ is the Kronecker's delta.


\subsection{DTWA method to the OTOC}
We briefly review the DTWA method for computing $\sigma_i^a(t)=\mathrm{Tr}(\hat{\sigma}_i^a(t) \rho)  \eqqcolon \langle \hat{\sigma}_i^a(t) \rangle$~\cite{schachenmayer2015many-body}, where $i\in \Lambda$, $a \in \{x,y,z\}$, $\hat{\sigma}_i^{a} (t)=e^{\mathrm{i}\hat{H}t} \hat{\sigma}_i^{a} e^{-\mathrm{i}\hat{H}t}$, and $\rho$ is the density matrix of the system.
Here we take $\hbar=1$.
Let us introduce the phase-point operator $\hat{A} (\bm{s}_{\bm{\tau}})$:
\begin{equation}
    \hat{A}(\bm{s}_{\bm{\tau}}) = \prod_{k \in \Lambda} \left[ \frac{1}{2} \left( 1+ \bm{s}_{\tau_k} \hat{\bm{\sigma}}_k^\intercal \right) \right],
\end{equation}
where $\bm{s}_{\bm{\tau}}=(\bm{s}_{\tau_1},\ldots,\bm{s}_{\tau_N})$ and $\bm{\tau}=(\tau_1,\ldots,\tau_N)$ with $\tau_k \in \{(0,0),(0,1),(1,0),(1,1) \}$ denotes the points in the discrete phase space by $\bm{s}_{(0,0)}=(1,1,1)$, $\bm{s}_{(0,1)}=(-1,-1,1)$, $\bm{s}_{(1,0)}=(1,-1,-1)$, and $\bm{s}_{(1,1)}=(-1,1,-1)$~\cite{william1987wigner}.
The discrete Wigner function for $\rho$ is given by
\begin{equation}
    W_{\bm{\tau}} = \frac{1}{2^N}\langle \hat{A}(\bm{s}_{\bm{\tau}}) \rangle.
\end{equation}
Then DTWA approximates $\sigma_i^a(t)$ as
\begin{align}
    \sigma_i^a (t) =&\sum_{\bm{\tau}} W_{\bm{\tau}} \mathrm{Tr} (\hat{\sigma}_i^a(t) \hat{A}(\bm{s}_{\bm{\tau}})) \nonumber\\
    =& \sum_{\bm{\tau}} W_{\bm{\tau}} \mathrm{Tr} (\hat{\sigma}_i^a e^{-i \hat{H}t} \hat{A}(\bm{s}_{\bm{\tau}}) e^{i \hat{H}t}) \nonumber\\
    \approx& \sum_{\bm{\tau}} W_{\bm{\tau}} \mathrm{Tr} [\hat{\sigma}_i^a \hat{A}(\bm{s}(t; {\bm{\tau}}))]=\sum_{\bm{\tau}} W_{\bm{\tau}} s_{i}^a(t;{\bm{\tau}}),
\end{align}
where $\bm{s}(t;{\bm{\tau}})=(\bm{s}_1(t;{\bm{\tau}}),\ldots,\bm{s}_N(t;{\bm{\tau}}))$ are the solutions of the following classical equations of motions at time $t$:
\begin{equation}
\frac{d}{dt} \bm{s}_k(t;{\bm{\tau}})= -2 \bm{s}_k(t;{\bm{\tau}}) \times \left( \bm{h} + \sum_{\ell(\neq k)} \bm{s}_\ell(t;{\bm{\tau}})  \bm{J}_{k\ell}^\intercal \right),
\label{eq:classical}
\end{equation}
with initial conditions $\bm{s}(0;\bm{\tau})=\bm{s}_{\bm{\tau}}$.
Here, $\times$ denotes the cross product.

In this work, we focus on the $n$-th order time-correlation function, defined as
\begin{equation}
    F_{ij}^{(n)ab}(t) = \langle (\hat{\sigma}_i^{a} (t) \hat{\sigma}_j^{b}(0))^n \rangle,
    \label{eq:OTOC_main}
\end{equation}
where $a, b \in \{x,y,z\}$ and $i,j \in \Lambda$.
A $k$-point correlation function $\braket{\hat{A}_1(t_1)\ldots \hat{A_k}(t_k)}$ is referred to as (anti-)time ordered when $t_1 \geq \ldots \geq t_k$ $(t_1 \leq \ldots \leq t_k)$ and as out-of-time ordered otherwise.
For $n\geq 2$, the correlator includes an out-of-time-ordered sequence, such as $\braket{\hat{\sigma}_1^x(t)\hat{\sigma}_2^y(0)\hat{\sigma}_1^x(t)\hat{\sigma}_2^y(0)}$ for $n=2$.
The case of $n\geq 3$, called higher-order OTOCs, has also been studied in~\cite{haehl2017schwinger,roberts2017chaos,tsuji2018out} along with the conventional $n=2$ case.
Although there exist other definitions of the OTOC (called a regularized OTOC~\cite{maldacena2016bound} or a bipartite OTOC~\cite{tsuji2018out}), this study uses the statistical average provided by $\braket{\cdot}=\mathrm{Tr}[\cdot \rho]$. 

We provide the DTWA expression for the $n$-th order time-correlation function (see the derivation for Appendix~\ref{appendix:DTWA}): for odd $n$
\begin{equation}
    F_{ij}^{(n)ab}(t)\approx \sum_{\bm{\tau}} W_{\bm{\tau}} [s_{\tau_j}^b s_{i,1}^{(n)a}(\bm{\tau}) + \mathrm{i} (s_{i,2}^{(n)a}(\bm{\tau}) - s_{i,3}^{(n)a}(\bm{\tau})) ],
    \label{eq:odd_OTOC}
\end{equation}
and for even $n$
\begin{equation}
    F_{ij}^{(n)ab}(t)\approx \sum_{\bm{\tau}} W_{\bm{\tau}} [s_{\tau_j}^b s_{j,1}^{(n)b}(\bm{\tau}) + \mathrm{i} (s_{j,2}^{(n)b}(\bm{\tau}) - s_{j,3}^{(n)b}(\bm{\tau})) ].
    \label{eq:even_OTOC}
\end{equation}
Here, $\bm{s}_{m}^{(\ell)}(\bm{\tau})=(\bm{s}_{1,m}^{(\ell)}(\bm{\tau}),\ldots,\bm{s}_{N,m}^{(\ell)}(\bm{\tau}))$ is obtained by $\bm{s}_{m}^{(\ell-1)}(\bm{\tau})$ for $m \in \{1,2,3\}$ and $\ell \in \{1,\ldots, n\}$: for odd $\ell$,
\begin{equation}
    \left\{
    \begin{aligned}
        &\bm{s}_{k,m}^{(\ell)}(\bm{\tau}) = \bm{s}_{k,m}^{(\ell-1)} (t;\bm{\tau}) \text{ for } k \neq i\\
        &\bm{s}_{i,m}^{(\ell)}(\bm{\tau}) = -\bm{s}_{i,m}^{(\ell-1)} (t;\bm{\tau}) + 2s_{i,m}^{(\ell-1)a}(t;\bm{\tau}) \mathbf{e}_a,
    \end{aligned}
    \right.
\end{equation}
and for even $\ell$,
\begin{equation}
    \left\{
    \begin{aligned}
        &\bm{s}_{k,m}^{(\ell)}(\bm{\tau}) = \bm{s}_{k,m}^{(\ell-1)} (-t;\bm{\tau}) \text{ for } k \neq j\\
        &\bm{s}_{j,m}^{(\ell)}(\bm{\tau}) = -\bm{s}_{j,m}^{(\ell-1)} (-t;\bm{\tau}) + 2s_{j,m}^{(\ell-1)b}(-t;\bm{\tau}) \mathbf{e}_b,
    \end{aligned}
    \right.
\end{equation}
where $\mathbf{e}_a$ is the unit vector along $a$-axis and $\bm{s}_{m}^{(\ell-1)}(\pm t;\bm{\tau}) = (\bm{s}_{1,m}^{(\ell-1)}(\pm t;\bm{\tau}),\ldots,\bm{s}_{N,m}^{(\ell-1)}(\pm t;\bm{\tau}))$ are the solutions of classical equations of motions in Eq.~(\ref{eq:classical}) at time $\pm t$ with initial conditions $\bm{s}_{m}(0;\bm{\tau})=\bm{s}_{m}^{(\ell-1)}(\bm{\tau})$.
Finally, $\bm{s}_{m}^{(0)}(\bm{\tau})$ are given as
\begin{align}
    &\bm{s}_{k,m}^{(0)}(\bm{\tau})=\bm{s}_{\tau_k} \text{ for } k\neq j \text{ and } m \in \{1,2,3\}, \nonumber\\
    &\left\{
    \begin{aligned}
        &\bm{s}_{j,1}^{(0)}(\bm{\tau})=s_{\tau_j}^b \mathbf{e}_b\\
        &\bm{s}_{j,2}^{(0)}(\bm{\tau})=\mathbf{e}_b \times \bm{s}_{\tau_j}\\
        &\bm{s}_{j,3}^{(0)}(\bm{\tau})=\bm{0}.
    \end{aligned}
    \right.
\end{align}

We adopt two types of the density matrix: all-down spin state $\rho_\downarrow = \ket{0^N}\bra{0^N}$ and infinite-temperature state $\rho_0 = \hat{I}/2^N$.
Here, $\ket{0^N}=\prod_{k \in \Lambda} \ket{0_k}$, where $\hat{\sigma}_k^z \ket{0_k}=-\ket{0_k}$, and $\hat{I}$ is the identity operator.
Then, the discrete Wigner functions are given as
\begin{align}
    W_{\bm{\tau}}=& \prod_{k \in \Lambda} \Big[\frac{1}{2} \left( \delta_{\tau_k (1,0)} + \delta_{\tau_k (1,1)} \right) \Big] \text{ for } \rho=\rho_\downarrow, \nonumber\\
    W_{\bm{\tau}}=& \prod_{k \in \Lambda} \Big[ \frac{1}{4} ( \delta_{\tau_k (0,0)} + \delta_{\tau_k (0,1)} \nonumber\\
    & \qquad +\delta_{\tau_k (1,0)} + \delta_{\tau_k (1,1)} ) \Big] \text{ for } \rho=\rho_0,
\end{align}
respectively.
Since the dimension of the discrete phase space grows exponentially with $N$, we approximate Eqs.~(\ref{eq:odd_OTOC}) and~(\ref{eq:even_OTOC}) by sampling $\bm{\tau}$ according to the probability of $W_{\bm{\tau}}$.
We set the number of samples to $100$.
We performed computations five times for each parameter set in Sec.~\ref{secA:result} to benchmark the method, and obtained the mean and standard deviation from these five instances.
Since the standard deviation is consistently small, we carried out a single simulation per parameter set  when estimating the scrambling time in Sec.~\ref{secB:result}.

In the following section, we compare the DTWA-based approximated dynamics with exact dynamics.
We use the permutation symmetry of the system Hamiltonian at $\alpha=0$ to compute the exact dynamics of $F_{ij}^{(1)ab}(t)$ with both $\rho=\rho_\downarrow$ and $\rho=\rho_0$, and $F_{ij}^{(2)ab}(t)$ with $\rho=\rho_0$.
Since the dimension of the permutation-symmetry subspace scales as $O(N^3)$~\cite{sarkar1987optical, gegg2016efficient}, we can perform large-scale simulation with system sizes up to $N=80$ (see Appendix~\ref{appendix:permutation} for the method).
For other cases, we simulate the quantum dynamics within the Hilbert space of $N$ spins.
When $\rho=\rho_\downarrow$, we compute $F_{ij}^{(n)ab}(t) = \bra{0^N} (\hat{\sigma}_i^{a} (t) \hat{\sigma}_j^{b})^n \ket{0^N}$, which can be obtained by solving the Schr{\"o}dinger equation, i.e. $\pm \mathrm{i} d \ket{\phi}/dt= \hat{H} \ket{\phi}$, where $+$ and $-$ signs represent the forward and backward time evolutions, respectively.
When $\rho=\rho_0$, we compute $F_{ij}^{(n)ab}(t)$ by using a random state $\ket{\phi_m}$, which is uniformly sampled from the Haar state:
\begin{equation}
    F_{ij}^{(n)ab}(t) \approx \frac{1}{M} \sum_{m=1}^M \bra{\phi_m} (\hat{\sigma}_i^{a} (t) \hat{\sigma}_j^{b})^n \ket{\phi_m}.
\end{equation}
$M$ is the number of the samples, and set to $100$.

The DTWA method needs a computational cost of the order of $N^2$ for evaluating $F_{ij}^{(n)ab}(t)$, whereas exact methods generally require computational costs that scale exponentially with $N$.
Thus, the DTWA offers a promising approach for exploring dynamical features of large-sized systems even when $\alpha > 0$.

\subsection{Benchmark results}\label{secA:result}
\begin{figure}[t]
    \centering
    (a)\\
    \includegraphics[width=0.65\linewidth]{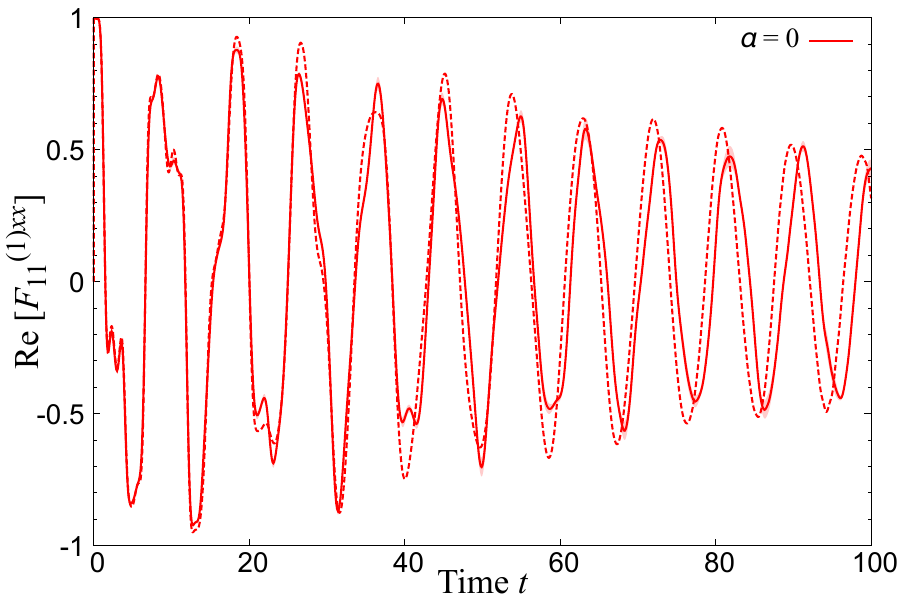}\\
    (b)\\
    \includegraphics[width=0.65\linewidth]{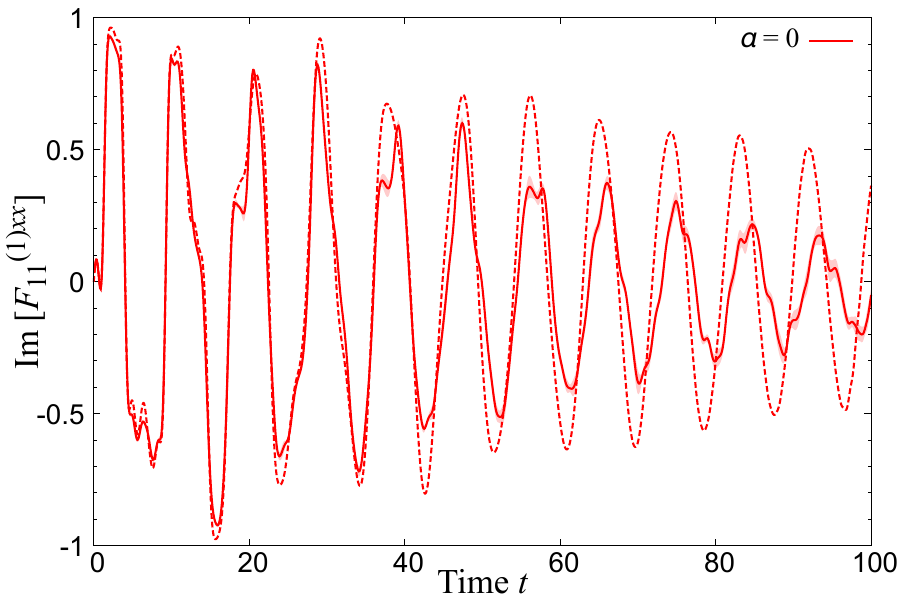}\\
    (c)\\   
    \includegraphics[width=0.65\linewidth]{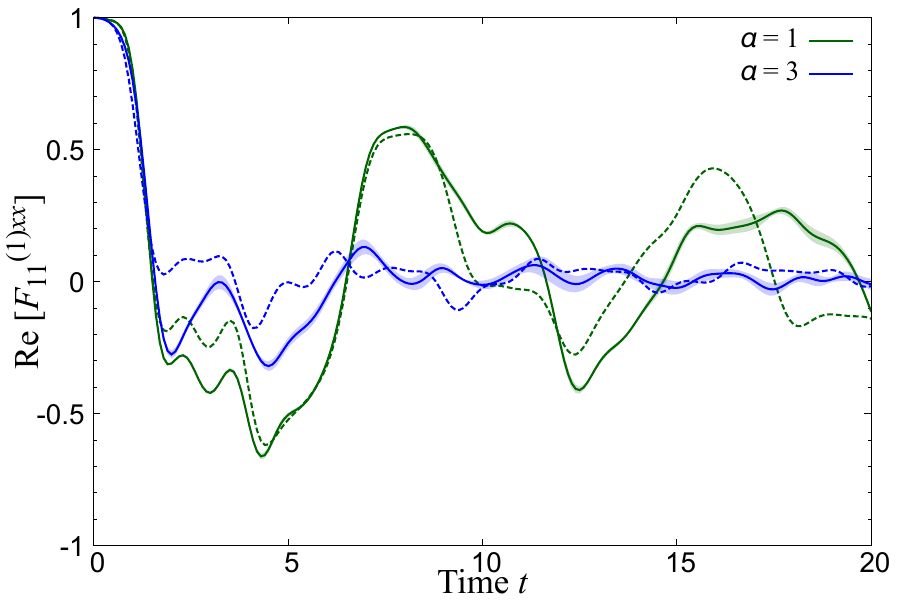}\\
    (d)\\
    \includegraphics[width=0.65\linewidth]{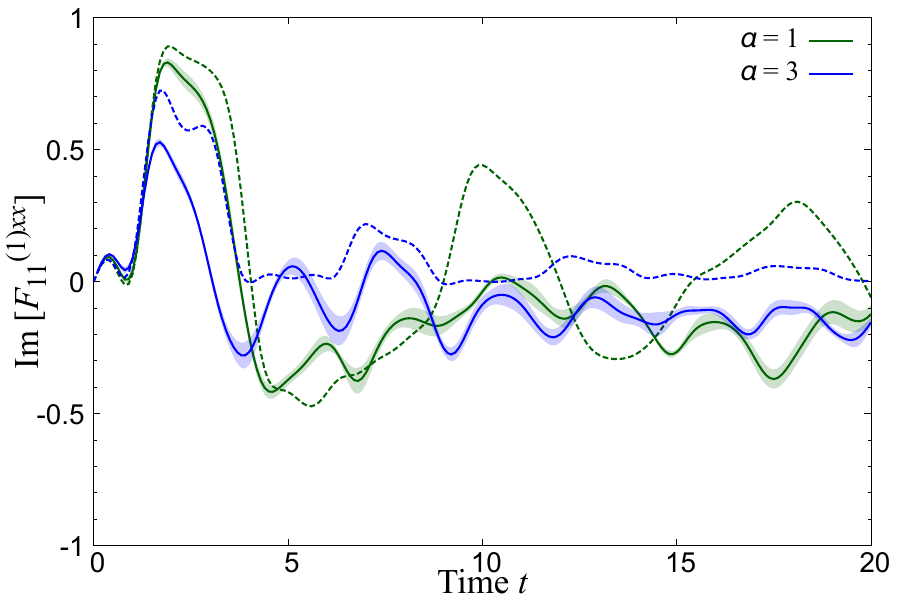}
    \caption{The real part and the imaginary part of autocorrelation functions.
    The DTWA results and exact results are drawn by bold lines and dotted lines, respectively.
    (a) and (b) show the data at $\alpha=0$ with $N=50$, whereas (c) and (d) show the data at $\alpha=1$ (green) and $\alpha=3$ (blue) with $N=15$.
    The shaded regions represent the standard deviation.
    }
    \label{fig:autocorrelation}
\end{figure}
We consider the following quantities: 
\begin{align}
    F_{ii}^{(1)aa}(t)=&\langle \hat{\sigma}_i^a(t)\hat{\sigma}_i^a \rangle, \nonumber\\
    C_{ij}^{(n)ab}(t) =& \frac{1}{2} \lVert [\hat{\sigma}_i^a(t),\hat{\sigma}_j^b] \rVert_{2n},
\end{align}
where $i,j \in \Lambda$, $a,b\in \{x,y,z\}$, $[\cdot,\cdot]$ denotes the commutator, and $\lVert \cdot \rVert_{2n}$ is the scaled Schatten $2n$-norms, defined by $\lVert \hat{O} \rVert_{2n}=(\langle (\hat{O}^\dagger \hat{O})^n \rangle)^{1/2n}$. 
The first quantity is the autocorrelation function, which plays a pivotal role in linear-response theory~\cite{kubo2012statistical} and Krylov complexity~\cite{parker2019universal}. 
The second quantity describes quantum information scrambling.
The Frobenius norm corresponding to $n=1$ and the operator norm corresponding to $n\to \infty$ were studied in previous works~\cite{tran2020hierarchy,colmenarez2020lieb,kuwahara2021absence}.
The Frobenius and operator norms represent the average and the fastest spreading of quantum information, respectively, and yield different lower bounds on $\gamma$~\cite{tran2020hierarchy,kuwahara2021absence}.

The scaled Schatten norm is related to the OTOC as
\begin{equation}
    C_{ij}^{(n)ab}(t)=\frac{1}{2}\left[\binom{2n}{n} \left(1+2\sum_{k=1}^n \frac{(-1)^k\binom{n}{k}}{\binom{n+k}{k}} F_{ij}^{(2k)ab}(t)\right)\right]^{\frac{1}{2n}},
    \label{eq:shatten}
\end{equation}
which is obtained by the binomial expansion of $(\hat{A}-\hat{B})^{2n}$, where $\hat{A}= \hat{\sigma}_i^a(t) \hat{\sigma}_j^b$ and $\hat{B}= \hat{\sigma}_j^b \hat{\sigma}_i^a(t)$, with properties of $\hat{A}\hat{B}=\hat{B}\hat{A}=\hat{I}$.
Eq.~(\ref{eq:shatten}) enables to study scrambling dynamics through the scaled Schatten norm.
For example, $n=1$ provides a relation involving the conventional OTOC:
\begin{equation}
    C_{ij}^{(1)ab}(t)=\sqrt{\frac{1- F_{ij}^{(2)ab}(t)}{2}}.
\end{equation}
Furthermore, the monotonicity property
\begin{equation}
    C_{ij}^{(n)ab}(t) \leq C_{ij}^{(n+1)ab}(t) \text{ for any } \rho,
    \label{eq:monotonicity}
\end{equation}
indicates that a larger $n$ gives a tighter lower bound for the operator norm.

\begin{figure}[t]
    \centering
    (a)\\
    \includegraphics[width=0.75\linewidth]{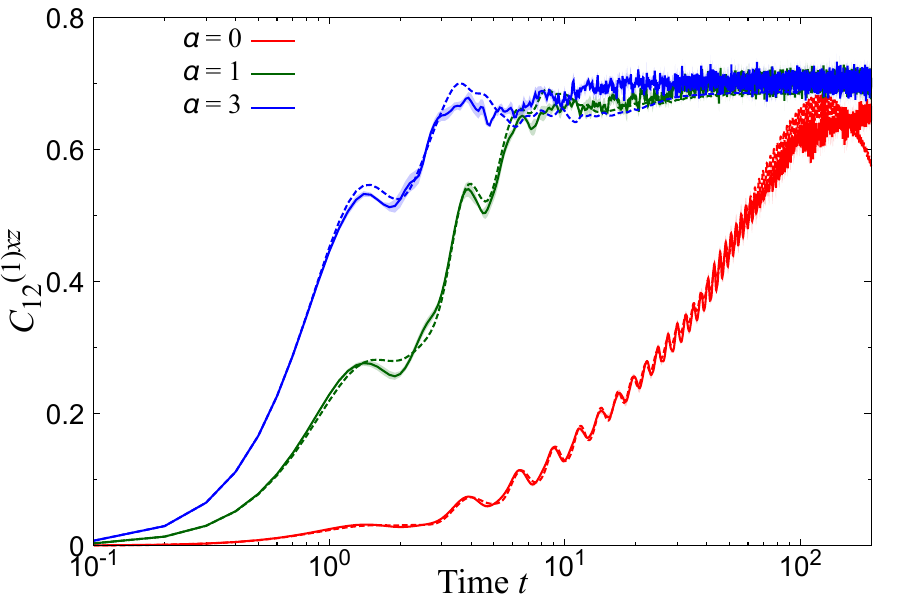}\\
    (b)\\
    \includegraphics[width=0.75\linewidth]{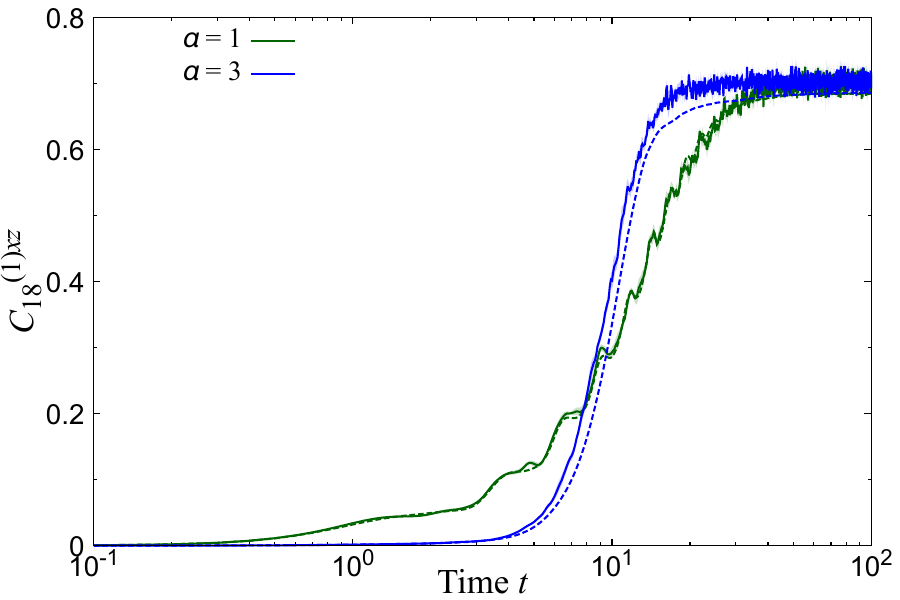}\\
    (c)\\
    \includegraphics[width=0.75\linewidth]{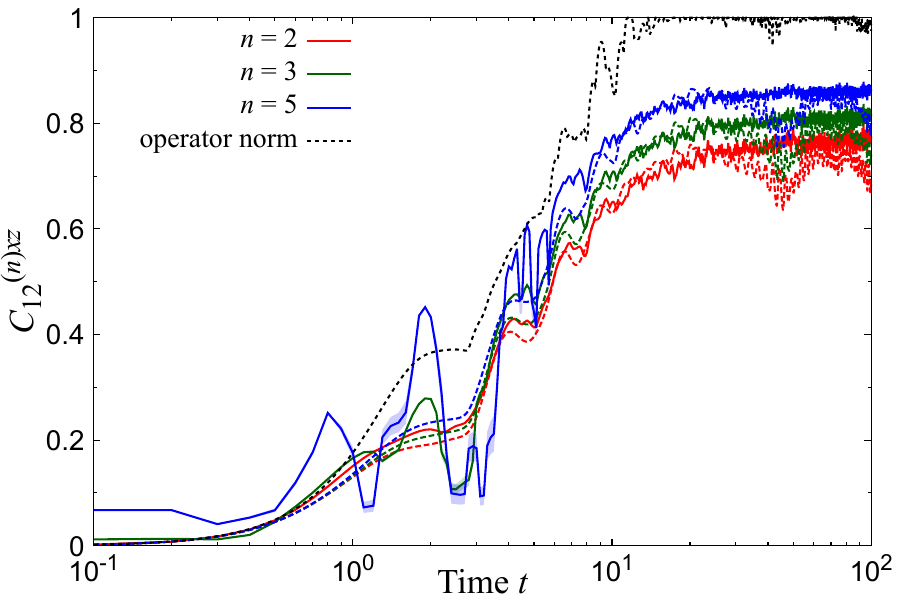}
    \caption{
    The OTOC results.
    The DTWA results and exact results are drawn by bold lines and dotted lines, respectively.
    (a) $C_{12}^{(1)xz}(t)$ and (b) $C_{1\lceil (N+1)/2 \rceil}^{(1)xz}(t)$ with $\rho=\rho_0$ are depicted for various values of $\alpha$: $\alpha=0$ and $N=50$ (red), $\alpha=1$ and $N=15$ (green), and $\alpha=3$ and $N=15$ (blue). 
    (c) $C_{12}^{(n)xz}(t)$ with $\rho=\rho_0$ at $n=2$ (red), $n=3$ (green), and $n=5$ (blue) at $\alpha=0$ and $N=10$ are depicted.
    The black dotted line represents the operator norm.
    The shaded regions represent the standard deviation.
    }
    \label{fig:otoc}
\end{figure}

Figure~\ref{fig:autocorrelation} depicts the real and imaginary parts of the autocorrelation function, $F_{11}^{(1)xx}(t)$, with the density matrix of $\rho=\rho_\downarrow$ for various values of $\alpha$ (see Appendix~\ref{appendix:size} for the system-size dependences of the results).
We find that the DTWA accurately reproduces the multiple oscillations of the exact dynamics for both real and imaginary parts at $\alpha=0$.
For all the cases, the DTWA captures the initial stage of the exact dynamics, but deviations from the exact dynamics appear earlier as $\alpha$ increases.
It is noted that a classical approach, where quantum fluctuations are entirely neglected and each spin has a definitive value in all directions, fails to capture even the initial dynamics.
In the case of $\rho=\rho_\downarrow$, the dynamics is described by a single trajectory $\bm{s}(t)=(\bm{s}_1(t),\ldots,\bm{s}_N(t))$, which obeys the classical equations of motions in Eq.~(\ref{eq:classical}) with initial conditions $(s_k^x(0), s_k^y(0), s_k^z(0))=(0,0,-1)$ for $k \in \Lambda$.
Then, we obtain
\begin{equation}
    F_{11}^{(1)xx}(t)
    \approx s_1^x (t) s_1^x(0) = 0.
\end{equation}
For a general $\rho$, the dynamics is described by an ensemble of trajectories.
However, simply taking an ensemble average of $s_1^x (t) s_1^x(0)$ always yields a real value, and thus fails to capture the imaginary part of the autocorrelation function, even though it is generally finite.
This result clearly shows that the DTWA method, which accounts for the quantum fluctuations in $\rho$, outperforms this classical approach.

Figures~\ref{fig:otoc}~(a) and~(b) illustrate the time evolution of the squared commutators, $C_{12}^{(1)xz}(t)$ and $C_{1\lceil (N+1)/2 \rceil}^{(1)xz}(t)$ with $\rho=\rho_0$, for various values of $\alpha$, where $\lceil \cdot \rceil$ is the ceiling function.
We do not depict $C_{1\lceil (N+1)/2 \rceil}^{(1)xz}(t)$ for $\alpha=0$, since $C_{12}^{(1)xz}(t)=C_{1\lceil (N+1)/2 \rceil}^{(1)xz}(t)$ in this case.
For all the cases, the squared commutators start at zero, increase with time, and saturate at late times.
The DTWA quantitatively reproduces the exact dynamics of $C_{12}^{(1)xz}(t)$ regardless of the value of $\alpha$.
While the DTWA accurately captures the dynamics of $C_{1\lceil (N+1)/2 \rceil}^{(1)xz}(t)$ at $\alpha=1$, it shows a noticeable deviation at $\alpha=3$.
These results indicate the validity of the DTWA method for estimating the average spreading of quantum information in the strongly long-range interacting systems (i.e., $0\leq \alpha \leq 1$).

To understand the limitation of the DTWA method, we examine a local field $\bm{\hat{\Gamma}}_i$, defined as
\begin{equation}
    \bm{\hat{\Gamma}}_i = \left( \sum_{j (\neq i)} \hat{\bm{\sigma}}_j \bm{J}_{ji} \right) + \bm{h}.
\end{equation}
In the DTWA method, this local field is approximated by a classical field.
For this approximation to be valid, the fluctuation of the local field, defined as $\langle \braket{(\hat{\Gamma}_i^a)^2}\rangle \coloneqq \braket{(\hat{\Gamma}_i^{a}-\braket{\hat{\Gamma}_i^{a}})^2}$, should be small.
This fluctuation can be roughly estimated as
\begin{align}
    \langle\braket{(\hat{\Gamma}_i^a)^2}\rangle &\sim \frac{1}{\mathcal{N}(\alpha)^{2}} \sum_{j (\neq i)} \sum_b \frac{\langle\braket{(\hat{\sigma}_j^b)^2}\rangle}{r_{ij}^{2\alpha}},
\end{align}
which scales as for large $N$
\begin{equation}
    \langle\braket{(\hat{\Gamma}_i^a)^2}\rangle \sim\left\{
    \begin{aligned}
        N^{-1} \quad &\text{ for $0\leq \alpha < 1/2$},\\
        N^{-1} \log N \quad &\text{ for $\alpha = 1/2$},\\
        N^{-2(1-\alpha)} \quad &\text{ for $1/2 < \alpha < 1$},\\
        (\log N)^{-2} \quad &\text{ for $\alpha = 1$},\\
        N^0 \quad &\text{ for $\alpha > 1$}.
    \end{aligned}
    \right.
\end{equation}
This scaling indicates that the fluctuation vanishes in the thermodynamic limit for $0\leq \alpha \leq 1$, thus supporting the validity of the DTWA method in this regime. 
In contrast, for $\alpha > 1$, the fluctuation remains to finite, leading to deviations from exact calculation.

Figure~\ref{fig:otoc}~(c) presents the time-evolution of $C_{12}^{(n)xz}(t)$ with $\rho=\rho_0$ for higher order $n \in \{2,3,5\}$ at $\alpha=0$.
The DTWA captures the saturated value of $C_{12}^{(n)xz}(t)$ at late times, but fails to describe the transient dynamics.
The deviations increase with the order $n$.
We find that $C_{12}^{(5)xz}(t)$ exhibits large fluctuations~\footnote{When the parenthesis in Eq.~(\ref{eq:shatten}) is negative in numerical simulations, the value of $C_{12}^{(5)ab}(t)$ is set to zero.}, and thus does not even qualitatively align with the exact dynamics.
The large fluctuations at $n=5$ remain for larger system size at $N=50$ (not shown).
Capturing $C_{12}^{(n)xz}(t)$ for large $n$ is difficult, since the parenthesis in Eq.~(\ref{eq:shatten}) at short-time scales is exponentially small with $n$.
The standard deviation of $C_{12}^{(n)xz}(t)$ is small compared to its mean, indicating that these deviations arise from the intrinsic limitations of the DTWA method, rather than insufficient sampling.
Furthermore, since $C_{12}^{(n)xz}(t)$ for a small $n$ significantly differs from $C_{12}^{(\infty)xz}(t)$, developing a method to estimate the fastest spreading of quantum information, represented by the operator norm $C_{ij}^{(\infty)ab}(t)$, remains a challenging open problem.

\section{Scrambling time}\label{sec:result}
\subsection{Theoretical bound}\label{secB:theory}
We provide a theoretical lower bound on the scrambling time $t_*$ in the strongly long-range interacting regime.
We extend the applicability of the argument in Ref.~\cite{yin2020bound} from the case $\alpha=0$ to the range $0\leq \alpha \leq 1$.

Let $\mathcal{H}$ denote the Hilbert space of a spin-$1/2$ model on a lattice $\Lambda$, and let $\mathcal{B}=\otimes_{i\in \Lambda} \mathcal{B}_i$ be the set of Hermite operators acting on $\mathcal{H}$.
For simplicity, we restrict our argument to a one-dimensional lattice, though generalization to higher dimensions is straightforward.
We denote elements of $\mathcal{B}$ as $|\hat{O})$ and define an inner product on $\mathcal{B}$ by
\begin{equation}
    (\hat{O}|\hat{O'})=\mathrm{Tr}(\hat{O}^\dagger \hat{O'} \rho_0).
\end{equation}

We define the average operator size as
\begin{equation}
S(\hat{O})=\sum_{j\in \Lambda} \frac{(\hat{O}|\mathbb{P}_j|\hat{O})}{(\hat{O}|\hat{O})},
\end{equation}
where $\mathbb{P}_j$ is a projection operator, defined as
\begin{equation}
    \mathbb{P}_j |\hat{O})=|\hat{O}) - \frac{1}{2}|\hat{I}_j \otimes \mathrm{Tr}_j \hat{O}),
\end{equation}
with $\mathrm{Tr}_j$ denoting the partial trace over site~$j$ and $\hat{I}_j$ is the identity operator on site~$j$.
This projection extracts the components of Pauli strings with support on site~$j$.

We define the scrambling time $t_*$ as the time when the average operator size exceeds a fraction $a$ of the system size $N$, for some $N$-independent constant $0<a<3/4$:
\begin{equation}
    t_*= \inf_{t\in \mathbb{R}^+} \left\{ \sup_{\hat{O}_i \in \mathcal{B}_i} S(\hat{O}_i(t)) > a N\right\}.
\end{equation}
The operator size is related to the OTOC as
\begin{equation}
    S(\hat{\sigma}_i^\alpha(t))=\frac{1}{4}\sum_{j \in \Lambda} \sum_{b \in \{x,y,z\}} \left(1-F_{ij}^{(2)ab}(t)\right),
\end{equation}
implying that $t_*$ roughly corresponds to the time when the OTOC between distant sites becomes $O(1)$.
We now state the result.
\begin{theorem}
Let $\Lambda \in \{1,\ldots, N\}$ be a one-dimensional lattice, and consider the quantum dynamics on $\mathcal{H}$ governed by Eq.~(\ref{eq:model}), with interaction strength specified by Eq.~(\ref{eq:coupling}).
Then, scrambling time $t_*$ satisfies
\begin{equation}
    t_* \gtrsim \left\{
    \begin{aligned}
    &N^{\frac{1}{2}} \text{ for } 0\leq \alpha < \frac{1}{2}\\
    &N^{\frac{1}{2}} (\log N)^{-\frac{1}{2}} \text{ for } \alpha = \frac{1}{2}\\
    &N^{1-\alpha} \text{ for } \frac{1}{2} < \alpha\leq 1.
    \end{aligned}
    \right.
    \label{eq:theorem}
\end{equation}
\end{theorem}

\begin{proof}
We provide the proof, following Ref.~\cite{yin2020bound}, by bounding the average operator size $S(\hat{O}_i(t))$.
Without loss of generality, we assume $(\hat{O}_i|\hat{O}_i)=1$.
By unitary, this also holds at time $t$: $(\hat{O}_i(t)|\hat{O}_i(t))=1$. Thus, we have
\begin{align}
    S(\hat{O}_i(t))=&\sum_{j\in \Lambda} (\hat{O}_i(t)|\mathbb{P}_j|\hat{O}_i(t)) \nonumber\\
    \leq& 1 + \sum_{j (\neq i)}^N (\hat{O}_i(t)|\mathbb{P}_j|\hat{O}_i(t)),
    \label{eq:proof1}
\end{align}
where we have used $(\hat{O}_i(t)|\mathbb{P}_i|\hat{O}_i(t))\leq 1$.

To bound the second term, we define the projection
\begin{equation}
    \mathbb{P}_{i^c}=1 - \prod_{j(\neq i)} (1-\mathbb{P}_j),
\end{equation}
which projects onto operators with support outside site~$i$.
Using the inequality of arithmetic and geometric means, we obtain
\begin{align}
    1-\mathbb{P}_{i^c} \leq & \left[\frac{1}{N-1} \sum_{j(\neq i)}^N (1-\mathbb{P}_{j}) \right]^{N-1} \nonumber\\
    \Leftrightarrow  1-\mathbb{P}_{i^c} \leq & \frac{1}{N-1} \sum_{j(\neq i)}^N (1-\mathbb{P}_{j}) \nonumber \\
    \Leftrightarrow  \sum_{j(\neq i)}^N \mathbb{P}_{j} \leq & (N-1) \mathbb{P}_{i^c}.
    \label{eq:approximation}
\end{align}
This inequality approximates operators with support outside site~$i$ by operators with operator size of $N-1$.
Then,
\begin{equation}
    S(\hat{O}_i(t)) \leq 1 +(N-1) (\hat{O}_i(t)|\mathbb{P}_{i^c}|\hat{O}_i(t)).
    \label{eq:proof2}
\end{equation}

Next, we evaluate $(\hat{O}_i(t)|\mathbb{P}_{i^c}|\hat{O}_i(t)) = \|\mathbb{P}_{i^c}|\hat{O}_i(t))\|_2^2$.
Let us define the Hamiltonian as $\hat{H}=\hat{H}_i + \hat{H}_\mathrm{int} + \hat{H}_{i^c}$, with
\begin{equation}
\hat{H}_i= \bm{h} \hat{\bm{\sigma}}_i^\intercal, \quad \hat{H}_\mathrm{int}= \sum_{j(\neq i)} \hat{\bm{\sigma}}_i \bm{J}_{ij} \hat{\bm{\sigma}}_j^\intercal,
\end{equation}
and define the associated Liouvillians  $\mathcal{L}$, $\mathcal{L}_i$, $\mathcal{L}_\mathrm{int}$, and $\mathcal{L}_{i^c}$ (e.g., $\mathcal{L}|\hat{O})=i|[\hat{H},\hat{O}])$).
By the Duhamel identity,
\begin{equation}
    e^{\mathcal{L}t} = e^{\mathcal{L}_i t} + \int_0^t ds e^{\mathcal{L}(t-s)} (\mathcal{L} - \mathcal{L}_i) e^{\mathcal{L}_i s}.
\end{equation}
Thus, 
\begin{align}
    |\hat{O}_i(t))=&e^{\mathcal{L}_i t} |\hat{O}_i) + \int_0^t ds e^{\mathcal{L}(t-s)} (\mathcal{L}_\mathrm{int} + \mathcal{L}_{i^c}) e^{\mathcal{L}_i s} |\hat{O}_i) \nonumber\\
    =& |\hat{O}_i^{(0)}(t)) + \int_0^t ds e^{\mathcal{L}(t-s)} \mathcal{L}_\mathrm{int} |\hat{O}_i^{(0)}(s)),
\end{align}
where $|\hat{O}_i^{(0)}(t))=e^{\mathcal{L}_i t} |\hat{O}_i)$.
Since $\mathbb{P}_{i^c}|\hat{O}_i^{(0)}(t)) = 0$, we have
\begin{equation}
    \mathbb{P}_{i^c}|\hat{O}_i(t))=\int_0^t ds \mathbb{P}_{i^c} e^{\mathcal{L}(t-s)} \mathcal{L}_\mathrm{int} |\hat{O}_i^{(0)}(s)).
\end{equation}
Then, we use the triangle inequality to obtain
\begin{align}
    \| \mathbb{P}_{i^c}|\hat{O}_i(t)) \|_2 \leq & \int_0^t ds \| \mathbb{P}_{i^c} e^{\mathcal{L}(t-s)} \mathcal{L}_\mathrm{int} |\hat{O}_i^{(0)}(s)) \|_2 \nonumber\\
    \leq & \int_0^t \| \mathcal{L}_\mathrm{int} |\hat{O}_i^{(0)}(s)) \|_2 \leq 2 t \| \hat{H}_\mathrm{int}  \|_2,
\end{align}
where we have used $\| \mathbb{P}_{i^c} |\hat{O})\|_2\leq \||\hat{O})\|_2$ and $\|e^{\mathcal{L}t} |\hat{O})\|_2 = \||\hat{O})\|_2$.
Thus, 
\begin{equation}
   (\hat{O}_i(t)|\mathbb{P}_{i^c}|\hat{O}_i(t)) \leq 4t^2 \| \hat{H}_\mathrm{int}  \|_2^2.
   \label{eq:proof3}
\end{equation}

Finally, we estimate the interaction strength.
Using the orthogonality $(\hat{\sigma}_i^a \hat{\sigma}_j^b | \hat{\sigma}_i^c \hat{\sigma}_k^d)=\delta_{jk} \delta_{ac} \delta_{bd}$ and Eq.~(\ref{eq:coupling}), we obtain for large $N$
\begin{align}
    \| \hat{H}_\mathrm{int}  \|_2^2 =& \sum_{j(\neq i)}^N \sum_{a,b \in \{x,y,z\}} (J_{ij}^{ab})^2\nonumber\\
    \sim&\left\{
    \begin{aligned}
        &N^{-1} \text{ for } 0\leq \alpha < \frac{1}{2}\\
        &N^{-1} \log N \text{ for }\alpha=\frac{1}{2}\\
        &N^{-2(1-\alpha)} \text{ for } \frac{1}{2}<\alpha \leq 1.
    \end{aligned}
    \right.
    \label{eq:proof4}
\end{align}

Combining Eqs.~(\ref{eq:proof2}),~(\ref{eq:proof3}), and~(\ref{eq:proof4}) estimates
\begin{align}
    S(O_i(t)) \leq& 1+4(N-1)t^2 \| \hat{H}_\mathrm{int}  \|_2^2 \nonumber\\
    \simeq &\left\{
    \begin{aligned}
        &t^2 \text{ for } 0\leq \alpha < \frac{1}{2}\\
        &(\log N) t^2 \text{ for }\alpha=\frac{1}{2}\\
        &N^{-1+2\alpha} t^2 \text{ for } \frac{1}{2}<\alpha \leq 1,
    \end{aligned}
    \right.
    \label{eq:bound}
\end{align}
which leads to the lower bound on $t_*$ in Eq.~(\ref{eq:theorem}).
\end{proof}

\subsection{DTWA result}\label{secB:result}
\begin{figure}[t]
    \centering
    (a)\\
    \includegraphics[width=0.75\linewidth]{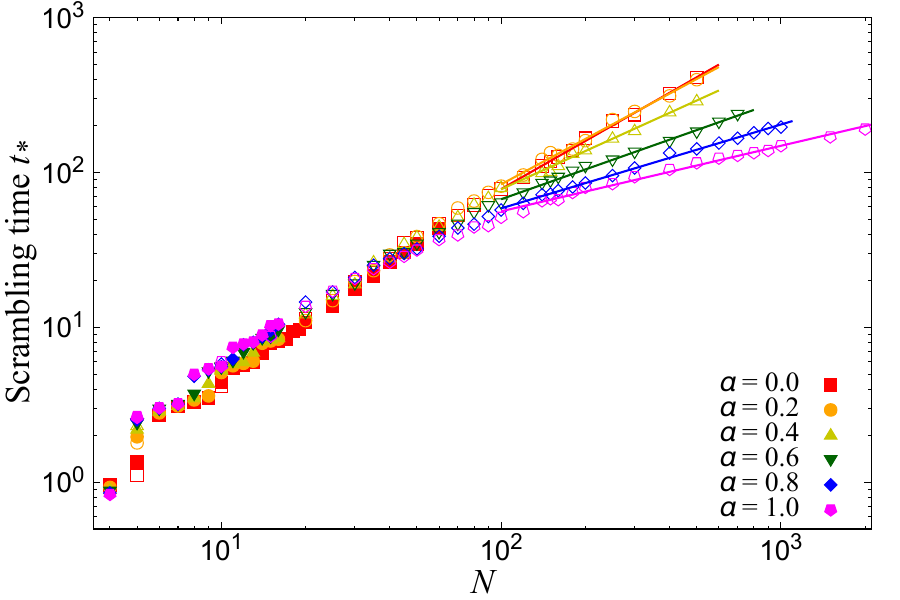}\\
    (b)\\
    \includegraphics[width=0.75\linewidth]{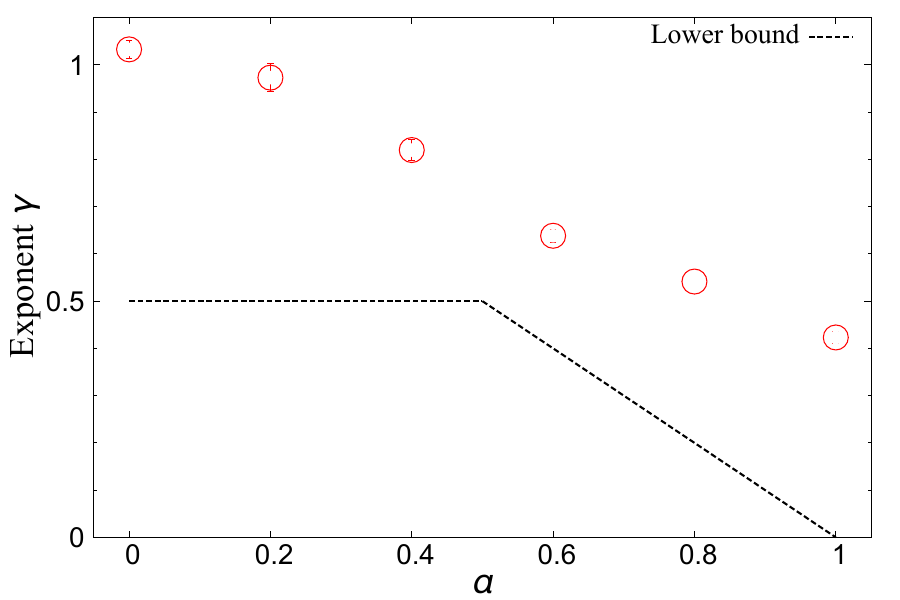}\\
    \caption{
    (a) System-size dependence of the scrambling time $t_*$ for different values of $\alpha$ in the strongly long-range interacting regime.
    Open symbols and filled symbols indicate the DTWA results and the exact results, respectively.
    The bold lines are fitting using data of $N\geq 100$.
    (b) $\alpha$-Dependence of the exponent $\gamma$ ($t_* \sim N^\gamma$). The error bars represent the standard deviation.
    }
    \label{fig:scrambling}
\end{figure}

We apply the DTWA method to investigate the system-size dependence of the scrambling time $t_*$ in the strongly long-range interacting systems (i.e., $0\leq \alpha \leq 1$).
To estimate $t_*$, we compute $t_i$ and $t_f$, where $t_i (t_f)$ is the time when $C^{(1)xz}_{1\lceil (N+1)/3 \rceil}=\sqrt{2}/4$ (i.e., $F^{(2)xz}_{1\lceil (N+1)/3 \rceil}=3/4$) is satisfied for the first (last) time.
The scrambling time is then given as the average: $t_*=(t_i+t_f)/2$.

Figure~\ref{fig:scrambling}~(a) shows the $N$-dependence of $t_*$ for various values of $\alpha$.
The DTWA and exact results are indicated by open and filled symbols, respectively.
For small system sizes, both results nearly overlap, demonstrating the validity of the DTWA method.
Unlike exact methods, DTWA is applicable to larger systems.
For large $N$, we reveal that the scaling behavior of $t_*$ with $N$ qualitatively changes depending on the value of $\alpha$.
It is noted that this dependence does not appear in small system sizes (it seems that results do not strongly depend on $\alpha$ for $N\lesssim 50$).
This fact implies that finite-size effects are strong in long-range interacting systems, and thus the DTWA is a powerful tool to study them.
Figure~\ref{fig:scrambling}~(b) displays the $\alpha$-dependence of the exponent $\gamma$, where $t_* \sim N^\gamma$.
The values of $\gamma$ are obtained by fitting the data for $N \geq 100$ in Figure~\ref{fig:scrambling}~(a).
The dashed line represents the lower bound for $\gamma$ given in Sec.~\ref{secB:theory}.
We find that the numerically estimated exponents exceed this lower bound.

To examine the tightness of the bound, we analyze the operator size of $\hat{\bm{\sigma}}_i(t)$ at short timescales, following Ref.~\cite{colmenarez2020lieb}.
A Taylor expansion yields
\begin{equation}
    \hat{\bm{\sigma}}_i(t)\approx\hat{\bm{\sigma}}_i-2t \hat{\bm{\sigma}}_i \times \left(\bm{h}+\sum_{j(\neq i)} \bm{J}_{ij} \hat{\bm{\sigma}}_j^\intercal \right),
\end{equation}
which approximates the operator size as 
\begin{equation}
    S(\hat{\sigma}_i^a(t)) \approx 1+4t^2 \|\hat{H}_\mathrm{int}^{\bar{a}}\|_2^2 \leq 1 + 4t^2 \|\hat{H}_\mathrm{int}\|_2^2,
\end{equation}
where $\hat{H}_\mathrm{int}^{\bar{a}}=\sum_{j(\neq i)} \sum_{b(\neq a)} \sum_c J_{ij}^{bc} \hat{\sigma}_i^b \hat{\sigma}_j^c$.
This indicates that the bound in Eq.~(\ref{eq:bound}) overestimates the operator size.
The overestimation arises from the inequality in  Eq.~(\ref{eq:approximation}), as shown by
\begin{align}
    &\sum_{j(\neq i)}(\hat{\sigma}_i^a(t)|\mathbb{P}_j |\hat{\sigma}_i^a(t)) \approx 4t^2 \|\hat{H}_\mathrm{int}^{\bar{a}}\|_2^2, \nonumber\\
    &(N-1)(\hat{\sigma}_i^a(t)|\mathbb{P}_{i^c} |\hat{\sigma}_i^a(t)) \approx 4(N-1)t^2 \|\hat{H}_\mathrm{int}^{\bar{a}}\|_2^2.
\end{align}
This difference occurs because local terms in $\hat{\sigma}_i^a(t)$, such as $\hat{\sigma}_i^x \hat{\sigma}_j^x$ for $j\neq i$, are approximated as operators with support over $N-1$ sites.

The DTWA results and the short-timescale analysis imply that the operator size would be overestimated at long timescales as well due to the inequality in Eq.~(\ref{eq:approximation}), and indicates that tighter bounds might exist.
Bridging the gap between numerical and analytical results remains an open problem for future work.

\section{Conclusion}\label{sec:conclusion}
This paper proposes a DTWA-based method for computing the OTOC.
The method is applied to analyze the time evolution of autocorrelation functions, squared commutators, and their higher-order extensions under Hamilton dynamics in long-range interacting systems.
By comparing the DTWA method with exact computations, we demonstrate that the DTWA accurately captures the average spreading of quantum information (i.e., squared commutators) across all time regimes in the strongly long-range interacting systems (Figs.~\ref{fig:otoc}~(a) and~(b)).
However, we also identify the limitations of the DTWA method in studying weakly long-range interacting systems and the fastest spreading of quantum information.
Next, we investigate the system-size dependence of the scrambling time, described by $t_*\sim N^\gamma$ (Figs.~\ref{fig:scrambling}~(a) and~(b)).
In order to numerically identify $\gamma$, data from large system sizes, which are not accessible by exact numerical computations, are needed.
It demonstrates strength of the DTWA method in studying long-range interacting systems.
Our results on $\gamma$ reveal that the numerically estimated value exceeds a theoretical lower bound.
It remains open to fill the gap between numerics and analytical results by deriving a tighter bound, if exists.
Additionally, extending the DTWA method to finite temperature and disordered systems will be an important direction for future work.

\begin{acknowledgments}
This work was supported by JSPS KAKENHI (Grant Numbers JP21H05185 and 23K13034) and by JST, PRESTO Grant No. JPMJPR2259.
The numerical calculations were partly supported by the supercomputer center of ISSP of Tokyo University.
\end{acknowledgments}


\appendix
\twocolumngrid

\section{DTWA expression for OTOC}~\label{appendix:DTWA}
We give a DTWA expression for the $n$-th order OTOC $F_{ij}^{(n)ab}(t)$.
In the discrete phase space, the $n$-th order OTOC is expressed as
\begin{equation}
    F_{ij}^{(n)ab}(t)= \sum_{\bm{\tau}} W_{\bm{\tau}} \mathrm{Tr} [(\hat{\sigma}_i^a(t) \hat{\sigma}_j^b)^n \hat{A}(\bm{s}_{\bm{\tau}})].
    \label{eq:OTOC0}
\end{equation}
Let us introduce
\begin{equation}
    G_{ij}^{(n)ab}(\bm{s}_{m}^{(0)}(\bm{\tau}))=\mathrm{Tr} [(\hat{\sigma}_i^a(t) \hat{\sigma}_j^b)^{n-1} \hat{\sigma}_i^a(t) \hat{A}(\bm{s}_{m}^{(0)}(\bm{\tau}))].
\end{equation}
Then the real part of the trace in $F_{ij}^{(n)ab}(t)$ is given as
\begin{align}
    & \mathrm{Re} \mathrm{Tr} [(\hat{\sigma}_i^a(t) \hat{\sigma}_j^b)^n \hat{A}(\bm{s}_{\bm{\tau}})] \nonumber\\
    =&\frac{1}{2} \mathrm{Tr} [(\hat{\sigma}_i^a(t) \hat{\sigma}_j^b)^{n-1} \hat{\sigma}_i^a(t) \{ \hat{\sigma}_j^b, \hat{A}(\bm{s}_{\bm{\tau}}) \}] \nonumber\\
    =&s_{\tau_j}^b \mathrm{Tr} [(\hat{\sigma}_i^a(t) \hat{\sigma}_j^b)^{n-1} \hat{\sigma}_i^a(t) \hat{A}(\bm{s}_{1}^{(0)}(\bm{\tau}))] \nonumber\\
    =&s_{\tau_j}^b G_{ij}^{(n)ab}(\bm{s}_{1}^{(0)}(\bm{\tau})),
    \label{eq:OTOC_real}
    \end{align}
where $\{\cdot, \cdot\}$ is the anticommutator,
and the imaginary part of the trace in $F_{ij}^{(n)ab}(t)$ is expressed as
\begin{align}
    & \mathrm{Im} \mathrm{Tr} [(\hat{\sigma}_i^a(t) \hat{\sigma}_j^b)^n \hat{A}(\bm{s}_{\bm{\tau}})] \nonumber\\
    =&\frac{-\mathrm{i}}{2} \mathrm{Tr} [(\hat{\sigma}_i^a(t) \hat{\sigma}_j^b)^{n-1} \hat{\sigma}_i^a(t) [ \hat{\sigma}_j^b, \hat{A}(\bm{s}_{\bm{\tau}}) ]] \nonumber\\
    =&\mathrm{Tr} [(\hat{\sigma}_i^a(t) \hat{\sigma}_j^b)^{n-1} \hat{\sigma}_i^a(t) (\hat{A}(\bm{s}_{2}^{(0)}(\bm{\tau})) - \hat{A}(\bm{s}_{3}^{(0)}(\bm{\tau})))] \nonumber\\
    =& G_{ij}^{(n)ab}(\bm{s}_{2}^{(0)}(\bm{\tau}))- G_{ij}^{(n)ab}(\bm{s}_{3}^{(0)}(\bm{\tau})).
    \label{eq:OTOC_imaginary}
\end{align}
For $n\geq 3$, we can show
\begin{align}
    &G_{ij}^{(n)ab}(\bm{s}_{m}^{(0)}(\bm{\tau}))\nonumber\\
    =& \mathrm{Tr} [\hat{\sigma}_j^b (\hat{\sigma}_i^a(t)\hat{\sigma}_j^b)^{n-2} \hat{\sigma}_i^a(t) \hat{A}(\bm{s}_{m}^{(0)}(\bm{\tau})) \hat{\sigma}_i^a(t)] \nonumber\\
    \approx&\mathrm{Tr} [\hat{\sigma}_j^b (\hat{\sigma}_i^a(t)\hat{\sigma}_j^b)^{n-2} e^{\mathrm{i}\hat{H}t}\hat{\sigma}_i^a \hat{A}(\bm{s}_{m}^{(0)}(t;\bm{\tau})) \hat{\sigma}_i^a e^{-\mathrm{i}\hat{H}t} ]\nonumber\\
    =&\mathrm{Tr} [\hat{\sigma}_j^b (\hat{\sigma}_i^a(t)\hat{\sigma}_j^b)^{n-2} e^{\mathrm{i}\hat{H}t} \hat{A}(\bm{s}_{m}^{(1)}(\bm{\tau})) e^{-\mathrm{i}\hat{H}t} ]\nonumber\\
    \approx&\mathrm{Tr} [ (\hat{\sigma}_i^a(t)\hat{\sigma}_j^b)^{n-3}  \hat{\sigma}_i^a(t) \hat{\sigma}_j^b\hat{A}(\bm{s}_{m}^{(1)} (-t;\bm{\tau})) \hat{\sigma}_j^b ] \nonumber\\
    =&\mathrm{Tr} [ (\hat{\sigma}_i^a(t)\hat{\sigma}_j^b)^{n-3}  \hat{\sigma}_i^a(t) \hat{A}(\bm{s}_{m}^{(2)}(\bm{\tau})) ] \nonumber\\
    =& G_{ij}^{(n-2)ab}(\bm{s}_{m}^{(2)}(\bm{\tau})).
\end{align}
By iteratively using this relation, we obtain for odd $n$
\begin{align}
        G_{ij}^{(n)ab}(\bm{s}_{m}^{(0)}(\bm{\tau})) =& G_{ij}^{(1)ab}(\bm{s}_{m}^{(n-1)}(\bm{\tau})) \nonumber\\
        \approx& \mathrm{Tr} [\hat{\sigma}_i^a \hat{A}(\bm{s}_{m}^{(n-1)}(t;\bm{\tau}))] \nonumber\\
        =& s_{i,m}^{(n-1)a}(t;\bm{\tau}) = s_{i,m}^{(n)a}(\bm{\tau}),
        \label{eq:G_odd}
\end{align}
and for even $n$
\begin{align}
        G_{ij}^{(n)ab}(\bm{s}_{m}^{(0)}(\bm{\tau})) =& G_{ij}^{(2)ab}(\bm{s}_{m}^{(n-2)}(\bm{\tau})) \nonumber\\
        \approx& \mathrm{Tr} [\hat{\sigma}_j^b e^{\mathrm{i}\hat{H}t}\hat{\sigma}_i^a \hat{A}(\bm{s}_{m}^{(n-2)}(t;\bm{\tau})) \hat{\sigma}_i^a e^{-\mathrm{i}\hat{H}t}] \nonumber\\
        =& \mathrm{Tr} [\hat{\sigma}_j^b e^{\mathrm{i}\hat{H}t} \hat{A}(\bm{s}_{m}^{(n-1)}(\bm{\tau})) e^{-\mathrm{i}\hat{H}t}] \nonumber\\
        \approx& \mathrm{Tr} [\hat{\sigma}_j^b \hat{A}(\bm{s}_{m}^{(n-1)} (-t;\bm{\tau}))] \nonumber\\
        =& s_{j,m}^{(n-1)b}(-t;\bm{\tau}) = s_{j,m}^{(n)b}(\bm{\tau}).
        \label{eq:G_even}
\end{align}
Combining Eqs.~(\ref{eq:OTOC0}),~(\ref{eq:OTOC_real}),~(\ref{eq:OTOC_imaginary}),~(\ref{eq:G_odd}), and~(\ref{eq:G_even}) obtains the DTWA expression for OTOCs in Eqs.~(\ref{eq:odd_OTOC}) and~(\ref{eq:even_OTOC}).

\section{Exact simulation method at $\alpha=0$}\label{appendix:permutation}
Here we explain the exact simulation method of the autocorrelation function $F_{ii}^{(1)aa}(t)$ with $\rho=\rho_\downarrow$ and $\rho=\rho_0$, and the OTOC $F_{ij}^{(2)ab}(t)$ with $\rho=\rho_0$.
The permutation symmetry of the Hamiltonian at $\alpha=0$ allows for the exact simulation even for large system sizes.
The OTOC was computed in~\cite{qi2023surprises}, but providing a detailed explanation here would be beneficial.
Following the argument in~\cite{gegg2016efficient}, we introduce a basis element as
\begin{equation}
    \ket{n_{11}, \bm{u}_{11}, n_{10}, \bm{u}_{10}} \bra{n_{11}, \bm{u}_{11}, n_{01}, \bm{u}_{01}},
\end{equation}
where $n_{ij}$ represents the number of spins that are represented by $\ket{i}\bra{j}$, where $\ket{1}$ and $\ket{0}$ are eigenvectors of $\hat{\sigma}^z$ with the eigenvalue of $1$ and $-1$, respectively, and $\bm{u}_{ij}$ is the set of spin labels.
For example, the basis denoted by $\ket{1,\{3\}, 0, \varnothing} \bra{1,\{3\},2,\{1,4\}}$ with $N=4$ represents a state, where spins~$1$ and~$4$ are in $\ket{0}\bra{1}$ state, spin~$3$ is in $\ket{1}\bra{1}$ state, and spin~$2$ is in $\ket{0}\bra{0}$ state.
Here, $\varnothing$ denotes the empty set.
Since $\bm{u}_{ij}$ is irrelevant within the subspace of the permutation symmetry, we define a basis $\hat{\Sigma}_{\vec{n}}$ as
\begin{equation}
    C_{\vec{n}}\hat{\Sigma}_{\vec{n}} = \sum_{\text{all comb.}} \ket{n_{11}, \bm{u}_{11}, n_{10}, \bm{u}_{10}} \bra{n_{11}, \bm{u}_{11}, n_{01}, \bm{u}_{01}},
    \label{eq:basis_all}
\end{equation}
where the sum runs over all possible sets of $\bm{u}_{11}$, $\bm{u}_{10}$, and $\bm{u}_{01}$, and $\vec{n}=(n_1,n_2,n_3)$ represents $n_1=n_{11}+n_{10}$, $n_2=n_{11}$, and $n_3=n_{01}$.
The normalization is given by
\begin{equation}
    C_{\vec{n}}=\frac{N!}{(n_1-n_2)!n_2!n_3!(N-n_1-n_3)!}.
\end{equation}
Here, we adopt the notation, where $n_1$, $n_2$, and $n_3$ corresponds to $n$, $m$, and $m'$ in Ref.~\cite{sarkar1987optical}, respectively.
Then, the operators that belong to the permutation-symmetry subspace can be expressed as
\begin{equation}
    \hat{A} = \sum_{n_1=0}^N \sum_{n_2=0}^{n_1} \sum_{n_3=0}^{N-n_1} A_{\vec{n}} \hat{\Sigma}_{\vec{n}},
\end{equation}
where $A_{\vec{n}} \in \mathbb{C}$ is a coefficient.
For example, the identity operator and all-spin-down state is expressed as $\hat{I}=\sum_{n_1=0}^N C_{(n_1,n_1,0)} \hat{\Sigma}_{(n_1,n_1,0)}$ and $\rho_\downarrow = \hat{\Sigma}_{(0,0,0)}$, respectively.
Eq.~(\ref{eq:basis_all}) gives the inner product of $\hat{\Sigma}_{\vec{n}}$ and $\hat{\Sigma}_{\vec{n}'}$ as
\begin{equation}
    \mathrm{Tr} [\hat{\Sigma}_{\vec{n}} \hat{\Sigma}_{\vec{n}'}] = \frac{1}{C_{\vec{n}}} \delta_{n'_1,n_2+n_3}\delta_{n'_2,n_2}\delta_{n'_3,n_1-n_2}.
    \label{eq:inner}
\end{equation}

To consider the time evolution of a spin operator $\hat{\bm{\sigma}}_i(t)=e^{\mathrm{i}\hat{H}t} \hat{\bm{\sigma}}_i e^{-\mathrm{i}\hat{H}t}$, we introduce a basis, where $\hat{\Sigma}_{\vec{n}}$ is sandwiched by spin operators, and expand $\hat{\bm{\sigma}}_i(t)$ using this basis as
\begin{equation}
    \hat{\bm{\sigma}}_i(t) = \sum_{p,q \in \{1,x,y,z\}} \sum_{\vec{n}} {\bm{\sigma}}^{(p,q)}_{\vec{n}}(t) \hat{\sigma}_i^p \hat{\Sigma}_{\vec{n}} \hat{\sigma}_i^q,
    \label{eq:single_spin}
\end{equation}
where $\hat{\sigma}_i^1$ is the identity operator and ${\bm{\sigma}}^{(p,q)}_{\vec{n}}(t) \in \mathbb{C}$ is a coefficient.
The Heisenberg equation $d\hat{\bm{\sigma}}_i(t)/dt = \mathrm{i}[\hat{H}, \hat{\bm{\sigma}}_i(t)]$ gives a set of differential equations for each element ${\bm{\sigma}}^{p,q}_{\vec{n}}(t)$.
The number of equations is the order of $N^3$.

There is a relation between the basis $\hat{\Sigma}_{\vec{n}}$ and $\sum_{i\in \Lambda}\hat{\sigma}_i^p \hat{\Sigma}_{\vec{n}} \hat{\sigma}_i^q$:
\begin{equation}
    \sum_{i \in \Lambda} \hat{\sigma}_i^p \hat{\Sigma}_{\vec{n}} \hat{\sigma}_i^q = \sum_{\vec{n}'} f_{\vec{n},\vec{n}'}^{p,q} \hat{\Sigma}_{\vec{n}'},
    \label{eq:sum_i}
\end{equation}
where $f_{\vec{n},\vec{n}'}^{p,q} \in \mathbb{C}$ is a coefficient.
For example, $f_{\vec{n},\vec{n}'}^{z,z}=[N-2(n_1-n_2+n_3)]\delta_{\vec{n},\vec{n}'}$.
Then, we can compute the following quantity using Eq.~(\ref{eq:inner}) and Eq.~(\ref{eq:sum_i}):
\begin{equation}
    g_{\vec{n},\vec{n}'}^{\vec{p},\vec{q}}=\sum_{i \in \Lambda} \mathrm{Tr} [(\hat{\sigma}_i^{p_1}\ldots \hat{\sigma}_i^{p_m}) \hat{\Sigma}_{\vec{n}}(\hat{\sigma}_i^{q_1}\ldots \hat{\sigma}_i^{q_m}) \hat{\Sigma}_{\vec{n}'}].
    \label{eq:g}
\end{equation}
Here, $\vec{p}=(p_1,\ldots,p_m)$ and $\vec{q}=(q_1,\ldots,q_m)$, where $m$ represents the length of the vectors, and $p_i, q_i \in \{1,x,y,z\}$.

First we consider an antocorrelation function: $F_{ii}^{(1)aa}(t)=\mathrm{Tr}[\hat{\sigma}_i^a (t) \hat{\sigma}_i^a \rho]$.
For $\rho=\rho_\downarrow$,
\begin{align}
    &F_{ii}^{(1)aa}(t) = \frac{1}{N} \sum_{k\in\Lambda} F_{kk}^{(1)aa}(t) \nonumber\\
    =&\frac{1}{N} \sum_{k\in\Lambda} \mathrm{Tr}[\hat{\sigma}_k^a (t) \hat{\sigma}_k^a \rho_\downarrow] \nonumber\\
    =&\frac{1}{N}\sum_{p,q} \sum_{\vec{n}} \sum_{k \in \Lambda} {\sigma}^{a(p,q)}_{\vec{n}}(t) \mathrm{Tr} [(\hat{\sigma}_k^p \hat{\Sigma}_{\vec{n}} \hat{\sigma}_k^q) (\hat{\sigma}_k^a \hat{\Sigma}_{\vec{0}})] \nonumber\\
    =&\frac{1}{N}\sum_{p,q} \sum_{\vec{n}} {\sigma}^{a(p,q)}_{\vec{n}}(t) g_{\vec{n},\vec{0}}^{(p,1),(q,a)}.
\end{align}
In the first line, we use translational symmetry.
From the second to the third line, we apply Eq.~(\ref{eq:single_spin}) and $\rho_\downarrow=\hat{\Sigma}_{\vec{0}}$.
In the last line, we use the expression in Eq.~(\ref{eq:g}).
Similarly, we obtain $F_{ii}^{(1)aa}(t)$ with $\rho= \rho_0$ by noting that $\rho_0= 2^{-N}\sum_{n_1=0}^{N} C_{(n_1,n_1,0)} \hat{\Sigma}_{(n_1,n_1,0)}$.

Next we consider an OTOC at infinite temperature state: $F_{ij}^{(2)ab}(t)=\mathrm{Tr}[\hat{\sigma}_i^a (t) \hat{\sigma}_j^b \hat{\sigma}_i^a (t) \hat{\sigma}_j^b]/2^N$ for $i\neq j$.
The permutation symmetry implies that
\begin{equation}
    F_{ij}^{(2)ab}(t)=\frac{1}{N(N-1)}\left[\sum_{k,\ell \in \Lambda} F_{k\ell}^{(2)ab}(t) - \sum_{k \in \Lambda} F_{kk}^{(2)ab}(t)\right].
    \label{eq:OTOC}
\end{equation}
The first term is given as
\begin{align}
    \sum_{k,\ell \in \Lambda} F_{k\ell}^{(2)ab}(t)=&\frac{1}{2^N} \sum_{k,\ell \in \Lambda} \mathrm{Tr}[\hat{\sigma}_k^a (t) \hat{\sigma}_\ell^b \hat{\sigma}_k^a (t) \hat{\sigma}_\ell^b].
    \label{eq:OTOC_part1}
\end{align}
Substituting into Eq.~(\ref{eq:OTOC_part1}) the following relation
\begin{align}
    &\sum_{\ell \in \Lambda} \hat{\sigma}_\ell^b \hat{\sigma}_k^a (t) \hat{\sigma}_\ell^b = \sum_{\ell \in \Lambda} \sum_{p,q} \sum_{\vec{n}} \sigma_{\vec{n}}^{(p,q)}(t) \hat{\sigma}_\ell^b \hat{\sigma}_k^p \hat{\Sigma}_{\vec{n}} \hat{\sigma}_k^q \hat{\sigma}_\ell^b \nonumber\\
    =& \sum_{p,q} \sum_{\vec{n}} \Big[ \sum_{\vec{n}'} \sigma_{\vec{n}}^{a(p,q)}(t) f_{\vec{n},\vec{n}'}^{b,b} \hat{\sigma}_k^p \hat{\Sigma}_{\vec{n}'} \hat{\sigma}_k^q \nonumber\\
    &-2\mathrm{i} \epsilon_{bpq} (\sigma_{\vec{n}}^{a(q,b)}(t) \hat{\sigma}_k^p \hat{\Sigma}_{\vec{n}} + \sigma_{\vec{n}}^{a(b,p)}(t) \hat{\Sigma}_{\vec{n}} \hat{\sigma}_k^q) \Big] \nonumber\\
    =&\sum_{p,q} \sum_{\vec{n}} h_{\vec{n}}^{ab(p,q)}(t) \hat{\sigma}_k^p \hat{\Sigma}_{\vec{n}} \hat{\sigma}_k^q,
\end{align}
where $h_{\vec{n}}^{ab(p,q)}(t) \in \mathbb{C}$ is a coefficient and $\epsilon_{bpq}$ is the Levi-Civita symbol, yields
\begin{align}
    \sum_{k,\ell \in \Lambda} F_{k\ell}^{(2)ab}(t)=&\frac{1}{2^N} \sum_{p,q,r,s} \sum_{\vec{n},\vec{n}'} h_{\vec{n}}^{ab(p,q)}(t) g_{\vec{n},\vec{n}'}^{(s,p),(q,r)} \sigma_{\vec{n}'}^{(r,s)}(t).
    \label{eq:OTOC1}
\end{align}
Similarly, we obtain
\begin{align}
    &\sum_{k \in \Lambda} F_{kk}^{(2)ab}(t)=\frac{1}{2^N} \sum_{k \in \Lambda} \mathrm{Tr}[\hat{\sigma}_k^a (t) \hat{\sigma}_k^b \hat{\sigma}_k^a (t) \hat{\sigma}_k^b] \nonumber\\
    =& \frac{1}{2^N}\sum_{p,q,r,s} \sum_{\vec{n},\vec{n}'} \sigma_{\vec{n}}^{a(p,q)}(t) \sigma_{\vec{n}'}^{a(r,s)}(t) g_{\vec{n},\vec{n}'}^{(s,b,p),(q,b,r)}.
    \label{eq:OTOC2}
\end{align}
Thus, substituting Eqs.~(\ref{eq:OTOC1}) and~(\ref{eq:OTOC2}) into Eq.~(\ref{eq:OTOC}) provides the OTOC $F_{ij}^{(2)ab}(t)$ with $\rho=\rho_0$.\\

\section{System-size dependences of DTWA method}\label{appendix:size}
\begin{figure*}[t]
    \centering
    Autocorrelation function $F_{11}^{(1)xx}(t)$ with $\rho=\rho_\downarrow$\\
    \includegraphics[width=0.32\linewidth]{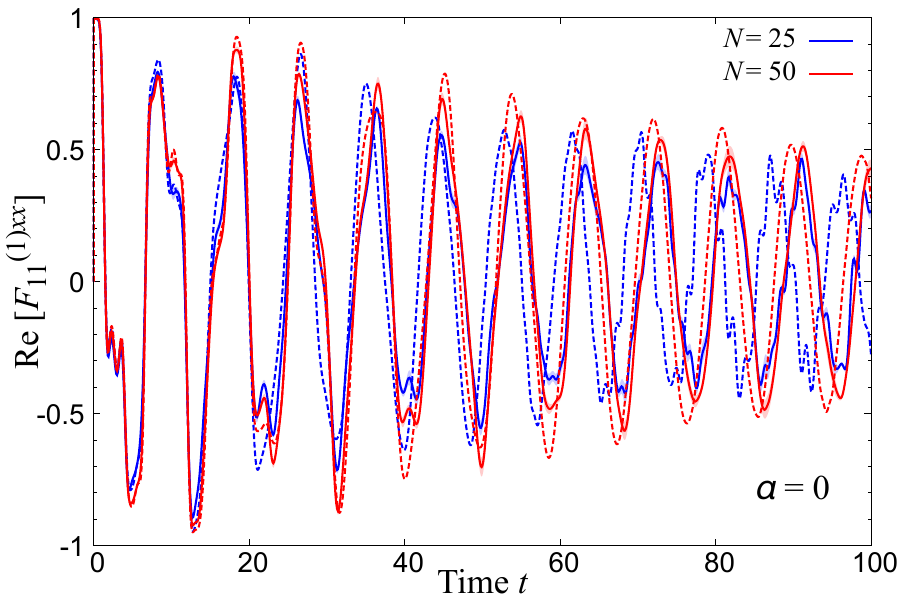}
    \includegraphics[width=0.32\linewidth]{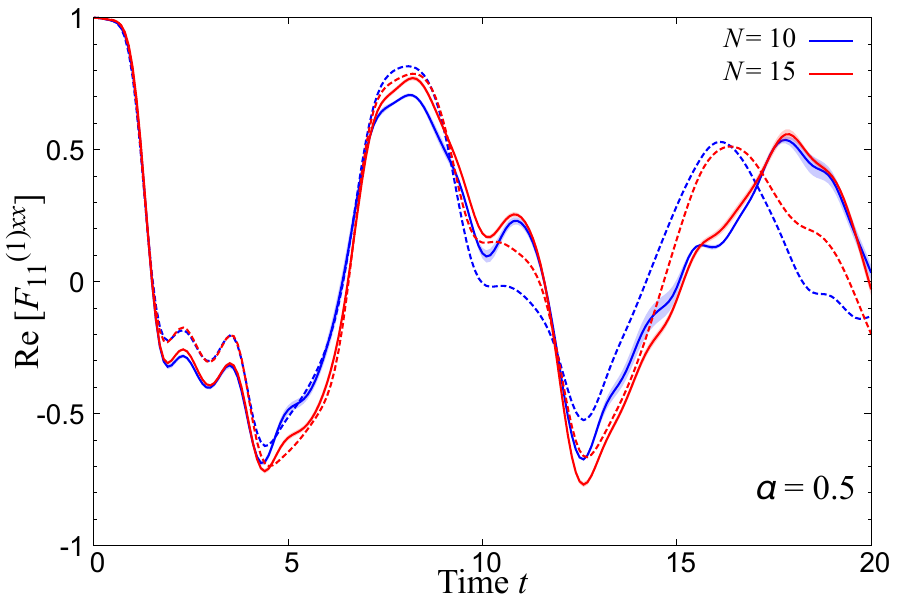}
    \includegraphics[width=0.32\linewidth]{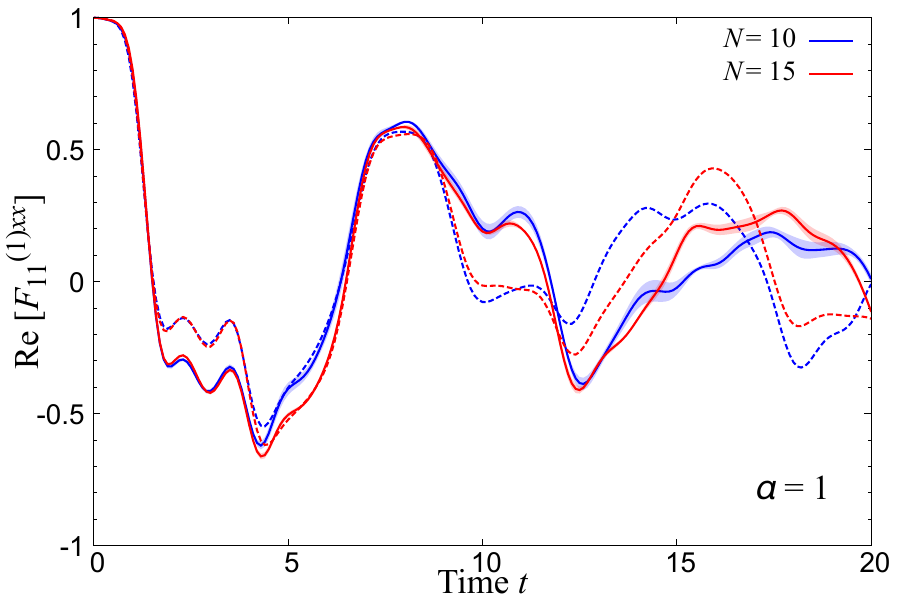}\\
    Autocorrelation function $F_{11}^{(1)xx}(t)$ with $\rho=\rho_0$\\
    \includegraphics[width=0.32\linewidth]{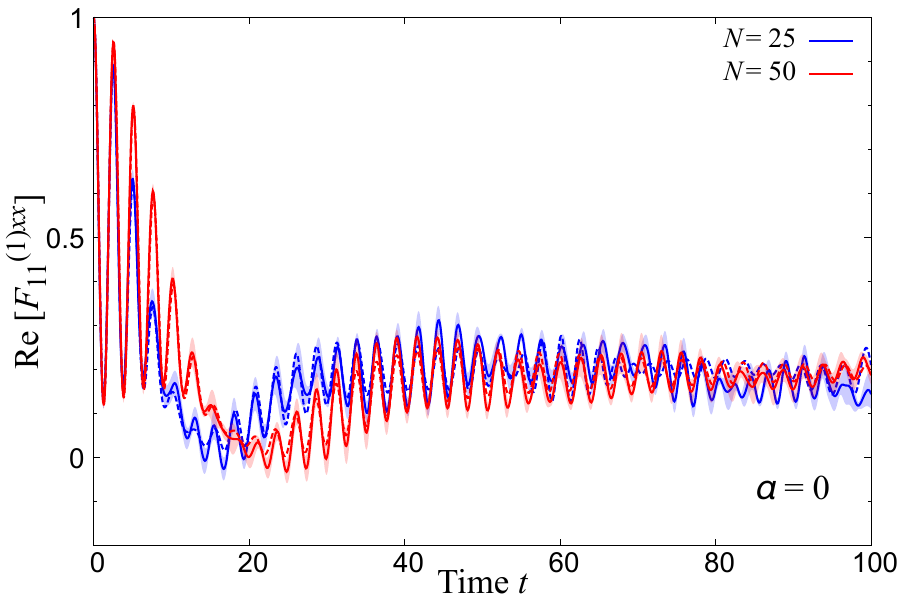}
    \includegraphics[width=0.32\linewidth]{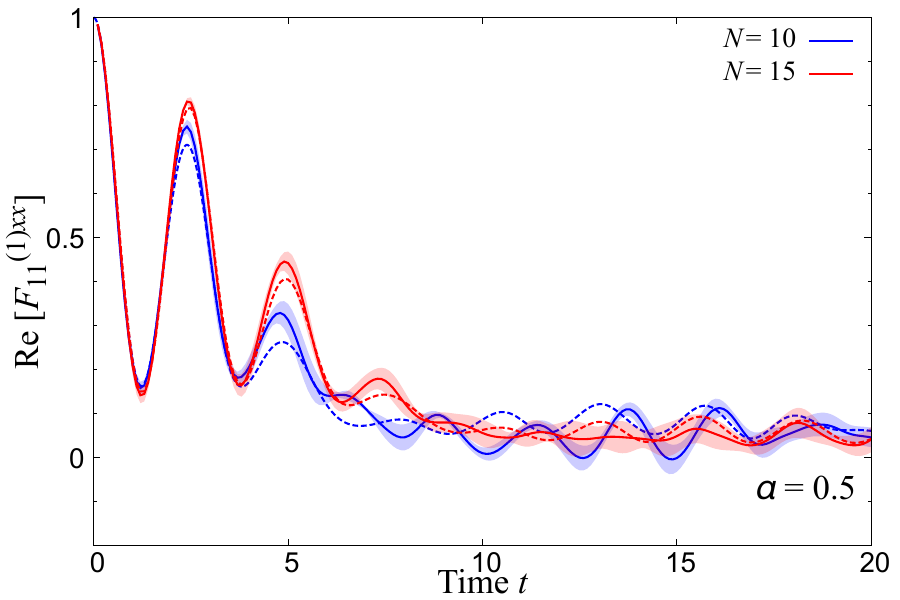}
    \includegraphics[width=0.32\linewidth]{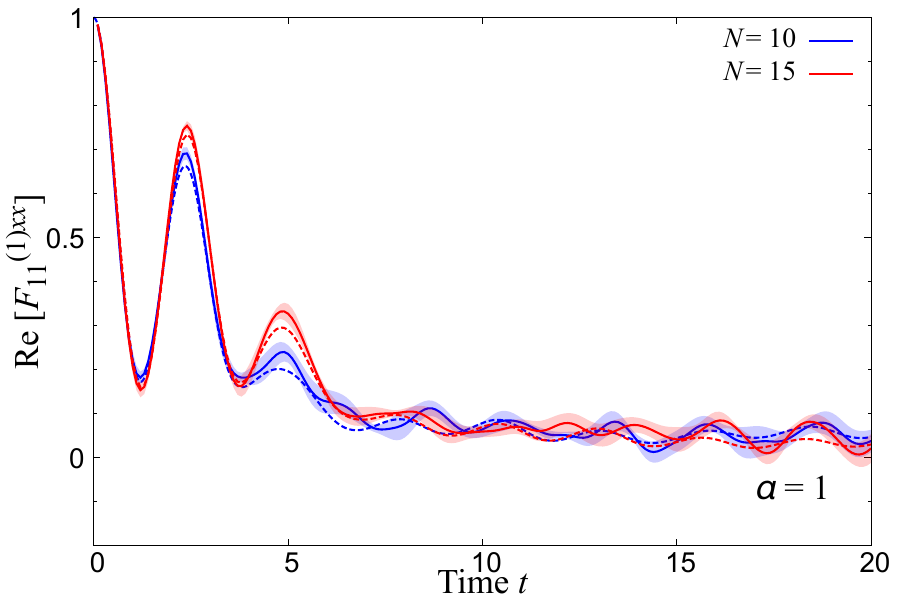}\\
    Frobenius norm $C_{12}^{(1)xz}(t)$\\
    \includegraphics[width=0.32\linewidth]{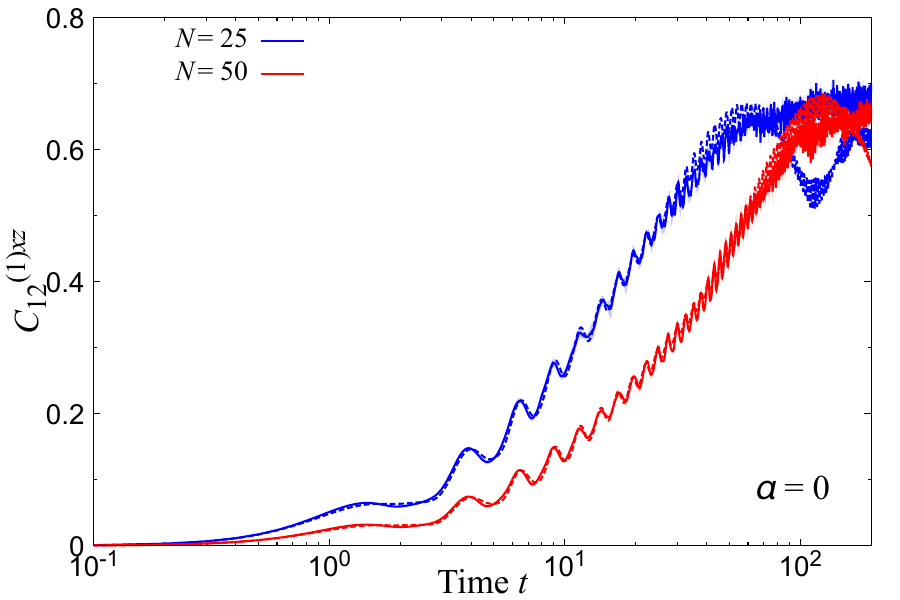}
    \includegraphics[width=0.32\linewidth]{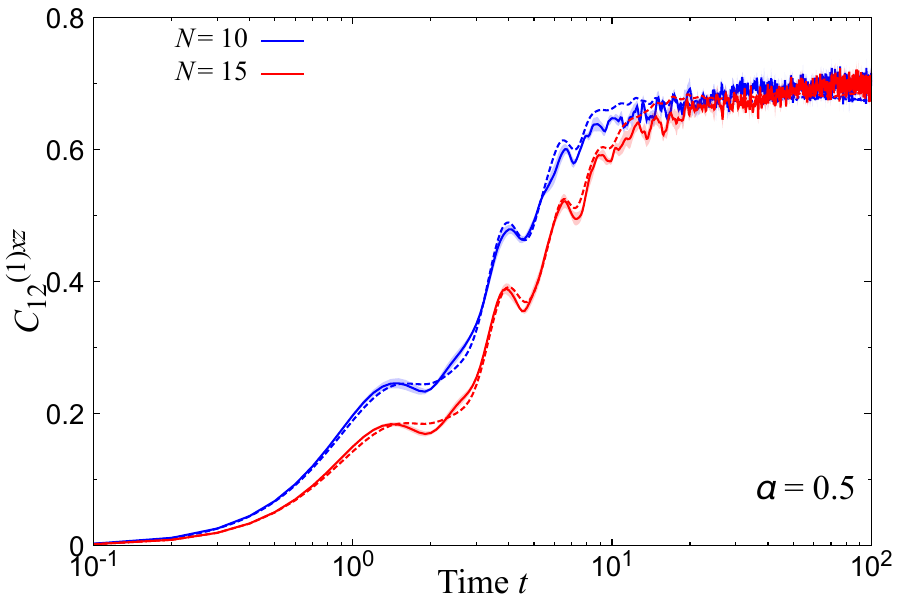}
    \includegraphics[width=0.32\linewidth]{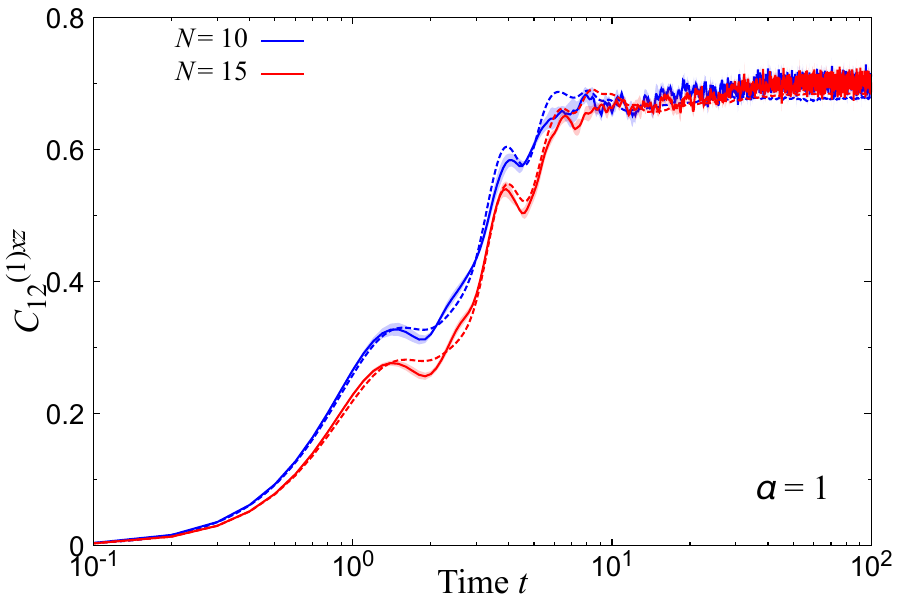}\\
    \caption{
    $F_{11}^{(1)xx}(t)$ with $\rho=\rho_\downarrow$ (upper panels), $F_{11}^{(1)xx}(t)$ with $\rho=\rho_0$ (middle panels), and $C_{12}^{(1)xz}(t)$  (lower panels) at various values of $\alpha \in \{0, 0.5, 1\}$.
    The red and blue lines represent the results for the larger and the smaller system sizes, respectively.
    The bold lines indicate the DTWA results, whereas the dotted lines correspond to the exact results.
    The shaded regions depict the standard deviation of the DTWA results.
    }
    \label{fig:system_size}
\end{figure*}

We present system-size dependences of the DTWA method in strongly long-range interacting systems at $\alpha \in \{0, 0.5, 1\}$.
The DTWA captures longer-time dynamics of $F_{11}^{(1)xx}(t)$ at $\alpha \in \{ 0, 0.5 \}$ with $\rho=\rho_\downarrow$ and $\rho_0$ as the system size increases.
On the other hand, at $\alpha=1$, where the interaction lies at the boundary, the DTWA reproduces the exact dynamics for almost the same duration across different system sizes.
For $C_{12}^{(1)xz}(t)$, the DTWA reproduces the exact dynamics even for small system sizes at all values of $\alpha$.
Additionally, at $\alpha=0$, we observe oscillations in $C_{12}^{(1)xz}(t)$ at late times, which are not captured by the DTWA method.

\bibliography{dtwa}


\end{document}